\documentclass[paper=letter]{article}
\usepackage{bm}
\usepackage{braket}
\usepackage[inline,shortlabels]{enumitem}
\usepackage[mathscr]{euscript}
\usepackage[utf8]{inputenc}
\usepackage{listings}
\usepackage{mathtools}
\usepackage{nccmath}
\usepackage{standalone}
\usepackage{subcaption}
\usepackage{todonotes}
\usepackage{xargs}
\usepackage{xfrac}
\usepackage{xspace}
 \usepackage[margin=1.0in]{geometry}
\usepackage{amsfonts}
\usepackage{amsthm}

\usepackage{csquotes}    

\usepackage{authblk}

\let\originalleft\left
\let\originalright\right
\renewcommand{\left}{\mathopen{}\mathclose\bgroup\originalleft}
\renewcommand{\right}{\aftergroup\egroup\originalright}

\makeatletter
\def\QEDbox{\begingroup\unitlength\p@\linethickness{.4\p@}\framebox(6,6){}\endgroup}

\makeatother

\makeatletter
\newcommand*{\myparagraph}[1]{\vskip12pt\noindent{\normalfont\normalsize\bfseries#1}.\,}
\makeatother

\usepackage{textcase}
\usepackage[T1]{fontenc}
\usepackage{lmodern}
\rmfamily
\DeclareFontShape{T1}{lmr}{bx}{sc} { <-> ssub * cmr/bx/sc }{}

\newcommandx*\disguisetext[3][1=c]{\makebox[\widthof{#2}][#1]{#3}}

\newcommandx*\disguisemath[3][1=c]{%
	\mathchoice{%
		\makebox[\widthof{$\displaystyle#2$}][#1]{$\displaystyle#3$}%
	}{%
		\makebox[\widthof{$\textstyle#2$}][#1]{$\textstyle#3$}%
	}{%
		\makebox[\widthof{$\scriptstyle#2$}][#1]{$\scriptstyle#3$}%
	}{%
		\makebox[\widthof{$\scriptscriptstyle#2$}][#1]{$\scriptscriptstyle#3$}%
	}%
}

\setlist{noitemsep}
\setlist[enumerate,1]{label=(\alph*)}
\setlist[enumerate,2]{label=(\roman*)}

\newtheorem{theorem}{Theorem}
\newtheorem{corollary}[theorem]{Corollary}
\newtheorem{lemma}[theorem]{Lemma}
\newtheorem{definition}[theorem]{Definition}

\newtheorem{claim}[theorem]{Claim}

\makeatletter
\let\@pklingPr\Pr
{\catcode`\|=\active
  \xdef\pr{\protect\expandafter\noexpand\csname pr \endcsname}
  \expandafter\gdef\csname pr \endcsname#1{\mathinner
        {\@pklingPr({\mathcode`\|32768\let|\midvert #1})}}
  \xdef\Pr{\protect\expandafter\noexpand\csname Pr \endcsname}
  \expandafter\gdef\csname Pr \endcsname#1{\@pklingPr\left(%
     \ifx\SavedDoubleVert\relax \let\SavedDoubleVert\|\fi
     {\let\|\SetDoubleVert
     \mathcode`\|32768\let|\SetVert
     #1}\right)}
}

{\catcode`\|=\active
  \xdef\ex{\protect\expandafter\noexpand\csname ex \endcsname}
  \expandafter\gdef\csname ex \endcsname#1{\mathinner
        {\mathbb{E}[{\mathcode`\|32768\let|\midvert #1}]}}
  \xdef\Ex{\protect\expandafter\noexpand\csname Ex \endcsname}
  \expandafter\gdef\csname Ex \endcsname#1{\mathbb{E}\left[%
     \ifx\SavedDoubleVert\relax \let\SavedDoubleVert\|\fi
     {\let\|\SetDoubleVert
     \mathcode`\|32768\let|\SetVert
     #1}\right]}
}

{\catcode`\|=\active
  \xdef\var{\protect\expandafter\noexpand\csname var \endcsname}
  \expandafter\gdef\csname var \endcsname#1{\mathinner
        {\operatorname{Var}[{\mathcode`\|32768\let|\midvert #1}]}}
  \xdef\Var{\protect\expandafter\noexpand\csname Var \endcsname}
  \expandafter\gdef\csname Var \endcsname#1{\operatorname{Var}\left[%
     \ifx\SavedDoubleVert\relax \let\SavedDoubleVert\|\fi
     {\let\|\SetDoubleVert
     \mathcode`\|32768\let|\SetVert
     #1}\right]}
}

{\catcode`\|=\active
  \xdef\cov{\protect\expandafter\noexpand\csname cov \endcsname}
  \expandafter\gdef\csname cov \endcsname#1{\mathinner
        {\operatorname{Cov}[{\mathcode`\|32768\let|\midvert #1}]}}
  \xdef\Cov{\protect\expandafter\noexpand\csname Cov \endcsname}
  \expandafter\gdef\csname Cov \endcsname#1{\operatorname{Cov}\left[%
     \ifx\SavedDoubleVert\relax \let\SavedDoubleVert\|\fi
     {\let\|\SetDoubleVert
     \mathcode`\|32768\let|\SetVert
     #1}\right]}
}
\makeatother

\DeclarePairedDelimiter\abs{\lvert}{\rvert}
\DeclarePairedDelimiterX{\norm}[1]{\lVert}{\rVert}{#1}

\DeclarePairedDelimiter\ceil{\lceil}{\rceil}
\DeclarePairedDelimiter\floor{\lfloor}{\rfloor}

\DeclarePairedDelimiter\intcc{[}{]}
\DeclarePairedDelimiter\intco{[}{)}

\newcommandx*{\LDAUOmicron}[2][1=@pkling_false]{\mathrm{O}\ifthenelse{\equal{#1}{small}}{\bigl(#2\bigr)}{\left(#2\right)}}
\newcommandx*{\LDAUomicron}[2][1=@pkling_false]{\mathrm{o}\ifthenelse{\equal{#1}{small}}{\bigl(#2\bigr)}{\left(#2\right)}}
\newcommandx*{\LDAUOmega}[2][1=@pkling_false]{\Omega\ifthenelse{\equal{#1}{small}}{\bigl(#2\bigr)}{\left(#2\right)}}
\newcommandx*{\LDAUomega}[2][1=@pkling_false]{\omega\ifthenelse{\equal{#1}{small}}{\bigl(#2\bigr)}{\left(#2\right)}}
\newcommandx*{\LDAUTheta}[2][1=@pkling_false]{\Theta\ifthenelse{\equal{#1}{small}}{\bigl(#2\bigr)}{\left(#2\right)}}

\newcommand*{\R}{\mathbb{R}}

\newcommand*{\N}{\mathbb{N}}

\newcommand{\disc}{\operatorname{disc}}
\newcommand{\Bin}{\operatorname{Bin}}

\newcommand{\balance}{\textsc{Balance}\xspace}
\newcommand{\shuffle}{\textsc{Shuffle}\xspace}
\newcommand{\vbased}{\textsc{Vertex-Based Balancer}\xspace}
\newcommand{\NRC}{\textsc{NRC}\xspace}

\newcommand{\cnt}[1]{c_{#1}}
\newcommand{\opinion}[1]{o_{#1}}
\newcommand{\majest}[1]{plu_{#1}}
\newcommand{\dom}[1]{dom_{#1}}
\newcommand{\estimate}[1]{e_{#1}}

\newcommand{\pmin}{p_{\text{min}}}
\newcommand{\tmix}{t_{\text{mix}}}

\newcommand{\mat}{\bm{M}}
\newcommand{\rw}{\bm{P}}

\newcommand{\textcite}{\cite}
\newcommand{\Textcite}{\cite}
\usepackage{hyperref}

\title{Plurality Consensus via Shuffling:\protect\\Lessons Learned from Load Balancing}

\author[1]{Petra Berenbrink}
\author[2]{Tom Friedetzky}
\author[1]{Peter Kling}
\author[1,3]{Frederik Mallmann-Trenn}
\author[2]{Chris Wastell}

\affil[1]{Simon Fraser University, Burnaby, Canada}
\affil[2]{Durham University, Durham, U.K.}
\affil[3]{École normale supérieure, Paris, France}

\date{}

\begin{document}
\maketitle
\begin{abstract}
We consider \emph{plurality consensus} in a network of $n$ nodes.
Initially, each node has one of $k$ opinions.
The nodes execute a (randomized) distributed protocol to agree on the \emph{plurality opinion} (the opinion initially supported by the most nodes).
Nodes in such networks are often quite cheap and simple, and hence one seeks protocols that are not only fast but also simple and space efficient.
Typically, protocols depend heavily on the employed communication mechanism, which ranges from sequential (only one pair of nodes communicates at any time) to fully parallel (all nodes communicate with all their neighbors at once) communication and everything in-between.

We propose a framework to design protocols for a multitude of communication mechanisms.
We introduce protocols that solve the plurality consensus problem and are with probability $1-\LDAUomicron{1}$ both time and space efficient.
Our protocols are based on an interesting relationship between plurality consensus and distributed load balancing.
This relationship allows us to design protocols that generalize the state of the art for a large range of problem parameters.
In particular, we obtain the same bounds as the recent result of \textcite{DBLP:conf/podc/AlistarhGV15} (who consider only two opinions on a clique) using a much simpler protocol that generalizes naturally to general graphs and multiple opinions.
\end{abstract}


\section{Introduction}
The goal of the \emph{plurality consensus} problem is to find the so-called \emph{plurality opinion} (i.e., the opinion that is initially supported by the largest subset of nodes) in a network $G$ where, initially, each of the $n$ nodes has one of $k$ opinions.
Applications for this problem include Distributed Computing~\cite{DBLP:conf/spaa/DoerrGMSS11,DBLP:journals/tcs/Peleg02,DBLP:conf/infocom/PerronVV09}, Social Networks~\cite{DBLP:conf/innovations/MosselS10,DBLP:journals/tcs/ClementiIGNS15,DBLP:journals/aamas/MosselNT14}, as well as biological interactions~\cite{chen2013programmable,CC14}.
All these areas typically demand both very simple and space-efficient protocols.
Communication models, however, can vary from anything between simple sequential communication with a single neighbor (often used in biological settings as a simple variant of asynchronous communication~\cite{DBLP:journals/eatcs/AspnesR07}) to fully parallel communication where all nodes communicate with all their neighbors simultaneously (like broadcasting models in distributed computing).
This diversity turns out to be a major bottleneck in algorithm design, since protocols (and their analysis) depend to a large part on the employed communication mechanism.

In this paper we present two simple plurality consensus protocols called \shuffle and \balance.
Both protocols work in a very general communication model which uses discrete rounds.
The communication partners are determined by a (possibly randomized) sequence $(\mat_t)_{t\leq N}$ of \emph{communication matrices}, where we assume\footnote{%
	For simplicity and without loss of generality; our protocols run in polynomial time in all considered models.
} $N$ to be some arbitrary large polynomial in $n$.
That is, nodes $u$ and $v$ can communicate in round $t$ if and only if $\mat_t[u,v]=1$.
In that case, we call the edge $\{u,v\}$ \emph{active}.
Our results allow for a wide class of communication patterns (which can even vary over time) as long as the communication matrices have certain \enquote{smoothing} properties (cf.~Section~\ref{sec:model}).
These smoothing properties are inspired by similar smoothing properties used by \textcite{DBLP:conf/focs/SauerwaldS12} for load balancing in the dimension exchange model.
In fact, load balancing is the source of inspiration for our protocols.
Initially, each node creates a suitably chosen number of tokens labeled with its own opinion.
Our \balance protocol then simply performs discrete load balancing on these tokens, allowing each node to get an estimate on the total number of tokens for each opinion.
The \shuffle protocol keeps the number of tokens on every node fixed, but shuffles tokens between communication partners.
By keeping track of how many tokens of their own opinion (label) were exchanged in total, nodes gain an estimate on the total (global) number of such tokens.
Together with a simple broadcast routine, the nodes are able to determine the plurality opinion.

The run time of our protocols is the smallest time $t$ for which all nodes have stabilized on the plurality opinion.
That is, all nodes have determined the plurality opinion and will not change.
This time depends on the network $G$, the communication pattern $(\mat_t)_{t\leq N}$, and the initial bias towards the plurality opinion (cf.~Section~\ref{sec:model}).
For both protocols we show a strong correlation between their run time and the mixing time of certain random walks and the (related) \emph{smoothing time}, both of which are used in the analysis of recent load balancing results~\cite{DBLP:conf/focs/SauerwaldS12}.
To give some more concrete examples of our results, let $T\coloneqq\LDAUOmicron{\log n/(1-\lambda_2)}$, where $1-\lambda_2$ is the spectral gap of $G$.
If the bias is sufficiently high, then both our protocols \shuffle and \balance determine the plurality opinion in time
\begin{enumerate*}
\item $n\cdot T$ in the \emph{sequential model} (only one pair of nodes communicates per time step);
\item $d\cdot T$ in the \emph{balancing circuit model} (communication partners are chosen according to $d$ (deterministic) perfect matchings in a round-robin fashion); and
\item $T$ in the \emph{diffusion model} (all nodes communicate with all their neighbors at once).
\end{enumerate*}
To the best of our knowledge, these match the best known bounds in the corresponding models.
For an arbitrary bias, the protocols differ in their time and space requirements.
More details about our results can be found in Section~\ref{subsec:our_contribution}.

\subsection{Related Work}
The subsequent discussion focuses on distributed models for large networks, where nodes are typically assumed to be very simple, and efficiency is measured in both time and space.
Most results depend on the \emph{initial bias} $\alpha\coloneqq\frac{n_1-n_2}{n}\in\intcc{\sfrac{1}{n},1}$, where $n_1$ and $n_2$ denote the number of nodes with the most common and second most common opinions, respectively.
The special case of plurality consensus with $k=2$ is often referred to as \emph{majority voting} or \emph{binary consensus}.
Similar to~\cite{DBLP:conf/soda/BecchettiCNPS15}, we use the term \emph{plurality} (instead of majority) to highlight that, for $k>2$, the opinion supported by the largest subset of nodes might be far from an (absolute) majority.

\myparagraph{Population Protocols}
The first major line of work on majority voting considers \emph{population protocols}.
Here, nodes are modelled as finite state machines with a small state space.
Communication partners are chosen either adversarial or randomly.
See~\cite{DBLP:journals/eatcs/AspnesR07,DBLP:journals/dc/AngluinAER07} for a more detailed model description.
\Textcite{DBLP:journals/dc/AngluinAE08} propose a 3-state (i.e., constant memory) population protocol for majority voting (i.e., $k=2$) on the clique to model the mixing behavior of molecules.
We refer to their communication model as the \emph{sequential model}: each time step, an edge is chosen uniformly at random, such that only one pair of nodes communicates.
If the initial bias $\alpha$ is $\LDAUomega{\log n/\sqrt{n}}$, their protocol lets all nodes agree (w.h.p.) on the majority opinion in $\LDAUOmicron{n\cdot\log n}$ steps.
\Textcite{DBLP:conf/icalp/MertziosNRS14} show that this 3-state protocol fails on general graphs in that there are infinitely many graphs on which it returns the minority opinion or has exponential run time.
They also provide a 4-state protocol for \emph{exact} majority voting, which \emph{always} returns the majority opinion (independent of $\alpha$) in time $\LDAUOmicron{n^6}$ on arbitrary graphs and in time $\LDAUOmicron{\frac{\log n}{\alpha}\cdot n^2}$ on the clique.
This result is optimal in that no population protocol for exact majority can have fewer than four states.
A very recent result is due to \textcite{DBLP:conf/podc/AlistarhGV15}.
They give a sophisticated (if slightly complicated) protocol for $k=2$ on the clique in the sequential model.
It solves exact majority and has (w.h.p.) \emph{parallel run time}\footnote{%
	Parallel run time in the sequential model is the number of (sequential) time steps divided by $n$.
	This is a typical measure for population protocols and based on the intuition that, in expectation, each node communicates with one neighbor within $n$ time steps.}
$\LDAUOmicron[small]{\frac{\log^2 n}{s\cdot\alpha}+\log^2 n\cdot\log s}$.
Here, $s$ is the number of states and must be in the range $s=\LDAUOmicron{n}$ and $s=\LDAUOmega{\log n\cdot\log\log n}$.

\myparagraph{Pull Voting}
The second major research line on plurality consensus has its roots in gossiping and rumor spreading.
Communication in these models is often restricted to pull requests, where nodes can query other nodes' opinions and use a simple rule to update their own opinion (note that the 3-state protocol from~\cite{DBLP:journals/dc/AngluinAE08} fits into this model).
See~\cite{DBLP:journals/tcs/Peleg02} for a slightly dated but thorough survey.
In a recent result, \textcite{DBLP:conf/icalp/CooperER14} consider a voting process for two opinions on arbitrary $d$-regular graphs.
They pull the opinion of two random neighbors and, if the pulled opinions are identical, adopt it.
For random $d$-regular graphs, (w.h.p.) all nodes agree after $\LDAUOmicron{\log n}$ steps on the plurality opinion (provided that $\alpha=\LDAUOmega[small]{\sqrt{1/d+d/n}}$).
For an arbitrary $d$-regular graph $G$, they need $\alpha=\LDAUOmega{\lambda}$ (where $1-\lambda_2$ is the spectral gap of $G$).
\Textcite{DBLP:conf/spaa/BecchettiCNPST14} consider a similar update rule on the clique for $k$ opinions.
Here, each node pulls the opinion of three random neighbors and adopts the majority opinion among those three (breaking ties uniformly at random).
They need $\LDAUOmicron{\log k}$ memory bits and prove a tight run time of $\LDAUTheta{k\cdot\log n}$ for this protocol (given a sufficiently high bias $\alpha$).
In another recent paper, \textcite{DBLP:conf/soda/BecchettiCNPS15} build upon the idea of the 3-state population protocol from~\cite{DBLP:journals/dc/AngluinAE08}.
Using a slightly different time and communication model, they generalize the protocol to $k$ opinions (on the clique).
In their model, nodes act in parallel and pull the opinion of a random neighbor each round.
Given a memory of $\log k+\LDAUOmicron{1}$ bits and assuming $k=\LDAUOmicron{(n/\log n)^{1/6}}$, they agree (w.h.p.) on the plurality opinion in time $\LDAUOmicron{k\cdot\log n}$ (given a sufficienctly high bias\footnote{%
	Their bias definition differs slightly from previous work, requiring $n_1\geq(1+\varepsilon)n_2$ for a constant $\varepsilon>0$.
}).
Note that, in contrast to all these results, we require our protocols to work for \emph{any} bias, even if it is only by one node (similar to~\cite{DBLP:conf/podc/AlistarhGV15}).

\myparagraph{Further Models}
Aside from the two research lines mentioned above, there is a multitude of related but quite different models.
They differ, for example, in the consensus requirement, the time model, or the graph models.
This paragraph gives merely a small overview over such model variants.
For details, the reader is referred to the corresponding literature.

In one very common variant of the voter model~\cite{HL75,DW83,DBLP:journals/iandc/HassinP01,DBLP:journals/siamdm/CooperEOR13,L85,AF14,LN07,Mal14}, one is interested in the time it takes for the nodes to agree on \emph{some} (arbitrary) opinion.
Notable representatives of this flavor are~\cite{DBLP:conf/spaa/DoerrGMSS11,BCNPT15}.
Both papers consider a consensus variant where the consensus can be on an arbitrary opinion (instead of on the plurality).
They have the additional requirement that the agreement is robust even in the presence of adversarial corruptions.
Another variant~\cite{DBLP:conf/infocom/PerronVV09} of distributed voting considers the 3-state protocol from~\cite{DBLP:journals/dc/AngluinAE08} (for two opinions on the complete graph), but in a continuous time model.
A third variant~\cite{DBLP:journals/dam/AbdullahD15} considers majority voting on special graphs given by a degree sequence.
Other protocols such as the one presented in~\cite{DBLP:journals/siamco/DraiefV12} guarantee converegence to the majority opinion.
The authors of~\cite{DBLP:journals/siamco/DraiefV12} analyse their protocol for 2 opinions.

\myparagraph{Load Balancing}
While our problem is quite different from load balancing, our results use techniques from and show interesting connections to certain types of discrete load balancers\footnote{%
	Similar connections can be drawn to work on averaging (e.g.,~\cite{DBLP:conf/focs/KDG03,DBLP:journals/tit/BoydGPS06}).
	However, due to the integrality/memory constraints, our problem is more closely related to (discrete) load balancing as considered in~\cite{RSW98,DBLP:conf/focs/SauerwaldS12}.
}.
In discrete load balancing, each node starts with an arbitrary number of tokens.
Each time step, nodes can exchange load over active edges.
The goal is typically to minimize the \emph{discrepancy} (the maximum load difference between any pair of nodes), and $K$ denotes the initial discrepancy.
The following results for $d$-regular graphs hold due to~\cite{DBLP:conf/focs/SauerwaldS12,RSW98}:
The discrepancy can be reduced to
\begin{enumerate*}
\item a constant in $\LDAUOmicron{n\cdot\log(Kn)/(1-\lambda_2)}$ time steps in the sequential model,
\item a constant in $\LDAUOmicron{d\cdot\log(Kn)/(1-\lambda_2)}$ time steps in the balancing circuits model for $d$ perfect matchings, and
\item $\LDAUOmicron{\sqrt{d\cdot\log n}/(1-\lambda_2)}$ in $\LDAUOmicron{\log(Kn)/(1-\lambda_2)}$ time steps in the diffusion model.
\end{enumerate*}

\subsection{Our Contribution}\label{subsec:our_contribution}
We introduce two protocols for plurality consensus, called \shuffle and \balance.
Both solve plurality consensus in a discrete time model under a diverse set of (randomized or adversarial) communication patterns for an arbitrary non-zero bias.
In particular, our very simple \balance protocol generalizes the work of \textcite{DBLP:conf/podc/AlistarhGV15} threefold:
\begin{enumerate*}
\item to arbitrary graphs,
\item to an arbitrary number $k$ of opinions, and
\item to more general (and truly parallel) communication models.
\end{enumerate*}
This generalization to parallel communication models requires a much more careful analysis, since we must deal with additional dependencies.
We continue with a detailed description of our results.
For both protocols, we give both run time and memory (measured in bits) bounds.

\myparagraph{\shuffle}
Our main result is the \shuffle protocol.
In the first time step each node generates $\gamma$ tokens labeled with its initial opinion.
During round $t$, any pair of nodes connected by an active edge (as specified by the communication pattern $(\mat_t)_{t\leq N}$) exchanges tokens.
We show that \shuffle solves plurality consensus and allows for a trade-off between run time and memory.
More exactly, let the number of tokens $\gamma=\LDAUOmicron{\log n/(\alpha^2\cdot T)}$, where $T$ is a parameter to control the trade-off between memory and run time.
Moreover, let $\tmix$ be such that any time interval $\intcc{t,t+\tmix}$ is \emph{$\varepsilon$-smoothing}\footnote{%
	Intuitively, this means that the communication pattern has good load balancing properties during any time window of length $\tmix$.
	This coincides with is the worst-case mixing time of a lazy random walk on active edges.
} (cf.~Section~\ref{sec:model}).
Then, \shuffle ensures that all nodes agree on the plurality opinion in $\LDAUOmicron{T\cdot\tmix}$ rounds (w.h.p.), using $\LDAUOmicron{\log n/(\alpha^2T)\cdot\log k+\log(T\cdot \tmix)}$ memory bits per node.
This implies, for example, that plurality consensus on expanders in the sequential model is achieved in $\LDAUOmicron{T\cdot n\log n}$ time steps and that nodes require only $\LDAUOmicron{\log n\cdot\log k/T+\log(Tn)}$ memory bits (assuming a constant initial bias).
For arbitrary graphs and many natural communication patterns (e.g., communicating with all neighbors in every round or communicating via random matchings), the time for plurality consensus is closely related to the spectral gap of the underlying communication network (cf.~Corollary~\ref{thm:main_corollary}).
To the best of our knowledge, this is the first plurality protocol (for more than two opinions) that can handle an arbitrary initial bias.

While our protocol is relatively simple, the analysis is much more involved.
The idea is to observe a single token only, and to show that, after $\tmix$ time steps, the token is (roughly) on any node with the same probability.
The main ingredients are Lemmas~\ref{lem:diff_majorized_by_rw} and~\ref{lem:neg_regression}, a generalization of a result by \textcite{DBLP:conf/focs/SauerwaldS12}.
These lemmas show that the joint distribution of token locations is negatively correlated, allowing us to derive a suitable Chernoff bound.
This is then used to show that, after $\tmix$ time steps, each node has a pretty good idea of the total numbers of tokens that are labeled with its opinion.
Using a broadcast-like protocol (nodes always forward their guess of the plurality), all nodes can determine the plurality opinion.
We believe that the non-trivial generalization of the negative correlation result is interesting in its own right.

\myparagraph{\balance}
The previous protocol (\shuffle) allows for a nice trade-off between run time and memory.
If the number of opinions is comparatively small, our much simpler \balance protocol gives better results.
In \balance, each node $u$ maintains a $k$-dimensional load vector.
If $j$ denotes $u$'s initial opinion, the $j$-th dimension of this load vector is initialized with $\gamma\in\N$ (a sufficiently large value) and any other dimension is initialized with zero.
In each time step, all nodes perform a simple, discrete load balancing on each dimension of these load vectors.
Our results imply, for example, that plurality consensus on expanders in the sequential model is achieved in only $\LDAUOmicron{n\cdot\log n}$ time steps with $\LDAUOmicron{k}$ memory bits per node (assuming a constant initial bias).
In the setting considered by~\textcite{DBLP:conf/podc/AlistarhGV15} (but for arbitrary $k$ instead of $k=2$), \balance achieves plurality consensus in time $\LDAUOmicron{n\cdot\log n}$ and uses $\LDAUOmicron{\log(1/\alpha)\cdot k}$ bits per node (Corollary~\ref{mastercor}).
This not only improves\footnote{%
	Note that the bounds in~\cite{DBLP:conf/podc/AlistarhGV15} are stated in parallel time, which is simply the normal run time divided by $n$.
} by a logarithmic factor over~\cite{DBLP:conf/podc/AlistarhGV15} (who consider $k=2$), but generalizes the results to $k>2$ via a much simpler protocol (although both protocols are similar in spirit).


\section{Model \& General Definitions}\label{sec:model}
We consider an undirected graph $G=(V,E)$ of $n\in\N$ nodes and let $1-\lambda_2$ denote the eigenvalue (or spectral) gap of $G$.
Note that $1-\lambda_2$ is constant for expanders and the clique.
Each node $u$ is assigned an \emph{opinion} $o_u\in\set{1,2,\dots,k}$.
For $i\in\set{1,2,\dots,k}$, we use $n_i\in\N$ to denote the number of nodes which have initially opinion $i$.
Without loss of generality (w.l.o.g), we assume $n_1>n_2\geq\dots\geq n_k$, such that $1$ is the opinion that is initially supported by the largest subset of nodes.
We also say that $1$ is the \emph{plurality opinion}.
The value $\alpha\coloneqq\frac{n_1-n_2}{n}\in\intcc{\sfrac{1}{n},1}$ denotes the \emph{initial bias} towards the plurality opinion.
In the \emph{plurality consensus problem}, the goal is to design simple, distributed protocols that let all nodes agree on the plurality opinion.
Time is measured in discrete rounds, such that the (randomized) run time of our protocols is the  number of rounds it takes until all nodes are aware of the plurality opinion.
As a second quality measure, we consider the total number of memory bits per node that are required by our protocols\footnote{%
	Literature more closely related to biological settings often uses the number of states $s\in\N$ a node can have to measure the space requirement of protocols (e.g.,~\cite{DBLP:conf/podc/AlistarhGV15}).
	Such protocols require $\ceil{\log s}$ bits per node.
}.
All our statements and proofs assume $n$ to be sufficiently large.

\myparagraph{Communication Model}
In any given round, two nodes $u$ and $v$ can communicate if and only if the edge between $u$ and $v$ is \emph{active}.
We use $\mat_t$ to denote the symmetric \emph{communication matrix} at time $t$, where $\mat_t[u,v]=\mat_t[v,u] =1 $ if $\{u,v\}$ is active and $\mat_t[u,v]=\mat_t[v,u]=0$ otherwise.
We assume (w.l.o.g) $\mat_t[u,u]=1$ (allowing nodes to \enquote{communicate} with themselves).
Typically, the sequence $\mat=(\mat_t)_{t\in\N}$ of communication matrices (the \emph{communication pattern}) is either randomized or adversarial, and our statements merely require that $\mat$ satisfies certain smoothing properties (see below).
For the ease of presentation, we restrict ourselves to polynomial number of time steps and consider only communication patterns $\mat=(\mat_t)_{t\leq N}$ where $N=N(n)$ is an arbitrarily large polynomial.
Let us briefly mention some natural and common communication models covered by such patterns:
\begin{itemize}
\item \emph{Diffusion Model:} All edges of the graph are permanently activated.
\item \emph{Random matching model:} In every round $t$, the active edges are given by a random matching.
	We require that random matchings from different rounds are mutually independent\footnote{%
		Note that there are several simple, distributed protocols to obtain such matchings~\cite{DBLP:journals/jcss/GhoshM96,DBLP:journals/tit/BoydGPS06}.
	}.
	While we do not restrict the exact way how the matching is chosen, results for the random matching model dependent on the parameter $\pmin\coloneqq\min_{t\in\N,\{u,v\}\in E}\Pr{\mat_t[u,v]=1}$.
\item \emph{Balancing Circuit Model:} There are $d$ perfect matchings $\mat_0,\mat_1,\dots,\mat_{d-1}$ given.
	They are used in a round-robin fashion, such that for $t\geq d$ we have $\mat_t=\mat_{t\bmod d}$.
\item \emph{Sequential Model:} In every round $t$, one edge $\{u,v\}\in E$ is chosen uniformly at random and activated (i.e., $\mat_t$ has exactly 4 non-zero entries).
\end{itemize}

\myparagraph{Notation}
We use $\norm{\bm{x}}_{\ell}$ to denote the $\ell$-norm of vector $\bm{x}$, where the $\infty$-norm is the vector's maximum absolute entry.
In general, bold font indicates vectors and matrices, and we use $x(i)$ to refer to the $i$-th component of vector $\bm{x}$.
The \emph{discrepancy} of a vector $\bm{x}$ is defined as $\disc(\bm{x})\coloneqq\max_ix(i)-\min_ix(i)$.
For a natural number $i\in\N$, we define $\intcc{i}\coloneqq\set{1,2,\dots,i}$ as the set of the first $i$ integers.
We use $\log x$ to denote the binary logarithm of $x\in\R_{>0}$.
We write $a \mid b$ if $a$ divides $b$.
For any node $u\in V$, we use $d(u)$ to denote $u$'s degree in $G$ and $d_t(u)\coloneqq\sum_v\mat_t[u,v]$ to denote its \emph{active degree} at time $t$ (i.e., its degree when restricted to active edges).
Similarly, $N(u)$ and $N_t(u)$ are used to refer $u$'s (active) neighborhood.
Moreover, let $\Delta\coloneqq\max_{t,u}d_t(u)$ be the maximum active degree of any node at any time in the given communication pattern.
We assume knowledge of $\Delta$, which merely means we assume that the nodes are aware of the communication model.

\myparagraph{Smoothing Property}\label{subsec:random_walk}
The run time of our protocols is closely related to the run time (\enquote{smoothing time}) of diffusion load balancing algorithms, which in turn is a function of the mixing time of a random walk on $G$.
More exactly, we consider a random walk on $G$ that is restricted to the active edges in each time step.
As indicated in Section~\ref{subsec:our_contribution}, this random walk should converge towards the uniform distribution over the nodes of $G$.
This leads to the following definition of the random walk's transition matrices $\rw_t$ based on the communication matrices $\mat_t$:
\begin{equation}
\rw_t[u,v]\coloneqq\begin{cases}
\frac{1}{2\Delta}        & \text{ if $\mat_t[u,v]=1$ and $u\not= v$,}\\
1-\frac{d_t(u)}{2\Delta} & \text{ if $\mat_t[u,v]=1$ and $u= v$,}    \\
0                        & \text{ if $\mat_t[u,v]=0$.}
\end{cases}
\end{equation}
Obviously, $\rw_t$ is doubly stochastic for all $t\in\N$.
Moreover, note that the random walk is trivial in any matching-based model, while we get $\rw_t[u,v]=\frac{1}{2d}$ for every edge $\{u,v\}\in E$ in the diffusion model on a $d$-regular graph.
We are now ready to define the required smoothing property.
\begin{definition}[$\varepsilon$-smoothing]
Consider a fixed sequence $(\mat_t)_{t\leq N}$ of communication matrices and a time interval $\intcc{t_1,t_2}$.
We say $\intcc{t_1,t_2}$ is \emph{$\varepsilon$-smoothing} (under $(\mat_t)_{t\leq N}$) if for any non-negative vector $\bm{x}$ with $\norm{\bm{x}}_{\infty}=1$ it holds that $\disc(\bm{x}\cdot\prod_{t=t_1}^{t_2}\rw_t)\leq\varepsilon$.
Moreover, we define the {\em mixing time} $\tmix(\varepsilon)$ as the smallest number of steps such that any time window of length $\tmix(\varepsilon)$ is $\varepsilon$-smoothing.
That is,
\begin{equation}
\tmix(\varepsilon)\coloneqq\min\set{t'|\forall t\in\N\colon\intcc{t,t+t'} \text{ is $\varepsilon$-smoothing}}
.
\end{equation}
\end{definition}
Note that the mixing time can be seen as the worst-case time required by a random walk to get \enquote{close} to the uniform distribution.
If the parameter $\varepsilon$ is not explicitly stated, we consider $\tmix\coloneqq\tmix(n^{-5})$.
To simplify the description of our protocols we assume that all nodes know $\tmix$.
This is without loss of generality, as we can \enquote{guess} the mixing time by a standard doubling-approach.
Note that $\tmix$ depends on the sequence $(\mat_t)_{t\leq N}$ of communication matrices.



\newcommand\shufflemem{$\left(12\cdot\frac{\log(n)}{\alpha^2\cdot T}+4\right)\cdot\log(k)+4\log\left(\frac{12\cdot\log(n)}{\alpha^2}\right)+\log(T\cdot\tmix)$\xspace}

\section{Protocol \shuffle}\label{sec:shuffle}
Our main result is the following theorem, stating the correctness as well as the time- and space-efficiency of \shuffle.
A formal description of \shuffle can be found in Section~\ref{subsec:protocol}, followed by its analysis in Section~\ref{standardvotermodel}.
\begin{theorem}\label{thm:main}
Let $\alpha = \frac{n_1-n_2}{n}\in\intcc{\sfrac{1}{n},1}$ denotes the initial bias.
Consider a fixed communication pattern $(\mat_t)_{t\leq N}$ and an arbitrary parameter $T\in\N$.
Protocol \shuffle ensures that all nodes know the plurality opinion after $\LDAUOmicron{T\cdot\tmix}$ rounds (w.h.p.)\footnote{We say an event happens with high probability (w.h.p.)if its probability is at least $1-1/n^c$ for $c\in\mathbb{N}$.} and requires \shufflemem memory bits per node.
\end{theorem}
The parameter $T$ in the theorem statement serves as a lever to trade run time for memory.
Since $\tmix$ depends on the graph and communication pattern, Theorem~\ref{thm:main} might look a bit unwieldy.
The following corollary gives a few concrete examples for common communication patterns on general graphs.
\begin{corollary}\label{thm:main_corollary}
Let $G$ be an arbitrary $d$-regular graph.
\shuffle ensures that all nodes agree on the plurality opinion (w.h.p.) using \shufflemem bits of memory in time
\begin{enumerate}
\item\label{thm:main_corollary:a} $\LDAUOmicron[small]{T\cdot\frac{\log(n)}{1-\lambda_2}}$  in the diffusion model,
\item\label{thm:main_corollary:b} $\LDAUOmicron[small]{\frac{T}{d\cdot\pmin}\cdot\frac{\log(n)}{1-\lambda_2}}$ in the random matching model,
\item\label{thm:main_corollary:c} $\LDAUOmicron[small]{T\cdot d\cdot\frac{\log(n)}{1-\lambda_2}}$  in the balancing circuit model, and
\item\label{thm:main_corollary:d} $\LDAUOmicron[small]{T\cdot n\cdot\frac{\log(n)}{1-\lambda_2}}$ in the sequential model.
\end{enumerate}
\end{corollary}

\subsection{Protocol Description}\label{subsec:protocol}
We continue to explain the \shuffle protocol given in Listing~\ref{shufflesimple}.
Our protocol consists of three parts that are executed in each time step:
\begin{enumerate*}
\item the \emph{shuffle} part,
\item the \emph{broadcast} part, and
\item the \emph{update} part.
\end{enumerate*}
Every node $u$ is initialized with $\gamma\in\N$ tokens labeled with $u$'s opinion $o_u$.
Our protocol sends $2 \Delta$ tokens chosen uniformly at random (without replacement) over each edge $\{u,v\} \in E$.
Here, $\gamma \geq 2 \Delta^2$ is a parameter depending on $T$ and $\alpha$ (to be fixed during the analysis).
\shuffle maintains the invariant that, at any time, all nodes have exactly $\gamma$ tokens.
In addition to storing the tokens, each node maintains a set of auxiliary variables.
The variable $\cnt{u}$ is increased during the update part and counts tokens labeled $\opinion{u}$.
The variable pair $(\dom{u},\estimate{u})$ is a temporary guess of the plurality opinion and its frequency.
During the broadcast part, nodes broadcast these pairs, replacing their own pair whenever they observe a pair with higher frequency.
Finally, the variable $\majest{u}$ represents the opinion currently believed to be the plurality opinion.
The shuffle and broadcast parts are executed in each time step, while the update part is executed only every
$\tmix$ time steps\footnote{%
	It is not essential for the protocol to know the mixing time.
	Using standard techniques, the protocol can guess the mixing time and grow the guess exponentially.
}.

Waiting $\tmix$ time steps for each update gives the broadcast enough time to inform all nodes and ensures that the tokens of each opinion are well distributed.
The latter implies that, if we consider a node $u$ with opinion $\opinion{u}=i$ at time $T\cdot\tmix$, the value $\cnt{u}$ is a good estimate of $T\cdot\gamma n_i/n$ (which is maximized for the plurality opinion).
When we reset the broadcast (Line~\ref{shufflesimple:line:reset}), the subsequent $\tmix$ broadcast steps ensure that all nodes get to know the pair $(\opinion{u},\cnt{u})$ for which $\cnt{u}$ is maximal.
Thus, if we can ensure that $\cnt{u}$ is a good enough approximation of $T\cdot\gamma n_i/n$, all nodes get to know the plurality.

\begin{lstlisting}[float,label={shufflesimple},caption={%
	Protocol \shuffle as executed by node $u$ at time $t$.
	At time zero, each node $u$ creates $\gamma$ tokens labeled $\opinion{u}$ and sets $\cnt{u}\coloneqq0$ and $(\dom{u},\estimate{u})\coloneqq(\opinion{u},\cnt{u})$.
}]
for $\{u,v\}\in E$ with $\mat_t[u,v]=1$:                     (*@\textbf{\{shuffle part\}}@*)
  send $2\Delta$ tokens chosen u.a.r. (without replacement) to $v$

for $\{u,v\}\in E$ with $\mat_t[u,v]=1$:                     (*@\textbf{\{broadcast part\}}@*)(*@\label{shufflesimple:line:broadcast_start}@*)
  send    $(\dom{u},\estimate{u})$
  receive $(\dom{v},\estimate{v})$
$v\coloneqq w$ with  $\estimate{w} \geq \estimate{w'}\quad\forall w,w'\in N_t(u)\cup\set{u}$
$(\dom{u},\estimate{u})\coloneqq(\dom{v},\estimate{v})$(*@\label{shufflesimple:line:broadcast_stop}@*)

if $t\equiv0\pmod\tmix$:                                    (*@\textbf{\{update part\}}@*)
  increase $c_u$ by the number of tokens labeled $\opinion{u}$ held by $u$
  $\majest{u}\coloneqq\dom{u}$            {plurality guess: last broadcast's dom. op.}
  $(\dom{u},\estimate{u})\coloneqq(\opinion{u},\cnt{u})$   {reset broadcast}(*@\label{shufflesimple:line:reset}@*)
\end{lstlisting}

\subsection{Analysis of \shuffle}\label{standardvotermodel}
Fix a communication pattern $(\mat_t)_{t\leq N}$ and an arbitrary parameter $T\in\N$.
Remember that $\tmix=\tmix(n^{-5})$ denotes the smallest number such that any time window of length $\tmix$ is $n^{-5}$-smoothing under $(\mat_t)_{t\leq N}$.
We set the number of tokens stored in each node to $\gamma\coloneqq\ceil{c\cdot\frac{\log n}{\alpha^2T}}$, where $c$ is a suitable constant.
The analysis of \shuffle is largely based on Lemma~\ref{onetwothree}, which states that, after $\LDAUOmicron{T\cdot\tmix}$ time steps, the counter values $\cnt{u}$ can be used to reliably separate the plurality opinion from any other opinion.
The main technical difficulty is the huge dependency between the tokens' movements, rendering standard Chernoff-bounds inapplicable.

Instead, we show that certain random variables satisfy the negative regression condition (Lemma~\ref{lem:neg_regression}), which allows us to majorize the token distribution by a random walk (Lemma~\ref{lem:diff_majorized_by_rw}) and to derive the following Chernoff bound.

\begin{lemma}\label{lem:chernoff_for_shuffle}
Consider any subset $B$ of tokens, a node $u\in V$, and an integer $T$.
Let $X\coloneqq\sum_{t\leq T}\sum_{j \in B}X_{j,t}$, where $X_{j,t}$ is $1$ if token $j$ is on node $u$ at time $t\cdot\tmix$.
With $\mu\coloneqq(1/n+1/n^5)\cdot\abs{B}\cdot T$, we have
\begin{equation}\label{huhuhu}
\Pr{X\geq(1+\delta)\cdot\mu}\leq e^{\delta^2\mu/3}
.
\end{equation}
\end{lemma}
The lemma's proof is a relatively straightforward consequence of Lemma~\ref{lem:diff_majorized_by_rw} (which is stated further below) and can be found at the end of this section.
Together, these lemmas generalize a result given in~\cite{DBLP:conf/focs/SauerwaldS12} to settings where nodes exchange load with more than one neighbor at a time, such that we have to deal with more complex dependencies.

\subsubsection*{Separating the Plurality via Chernoff}
Equipped with the Chernoff bound from Lemma~\ref{lem:chernoff_for_shuffle}, we prove concentration of the counter values and, subsequently, Theorem~\ref{thm:main}.
\begin{lemma}\label{onetwothree}
Let $c\geq 12$.
For every time $t\geq c\cdot T\cdot\tmix$ there exist values $\ell_{\top}>\ell_{\bot}$ such that
\begin{enumerate}
\item For all nodes $w$ with $\opinion{w}\geq2$ we have (w.h.p.) $\cnt{w}\leq\ell_{\bot}$.
\item For all nodes $v$ with $\opinion{\disguisemath{w}{v}}=1$ we have (w.h.p.) $\cnt{\disguisemath{w}{v}}\geq\ell_{\top}$.
\end{enumerate}
\end{lemma}
\begin{proof}
For two nodes $v$ and $w$ with $\opinion{v}=1$ and $\opinion{w}\geq 2$,
$\mu_i\coloneqq(1/n+1/n^5)c\cdot T\cdot\gamma\cdot n_k$ for all $i\in\intcc{k}$, and $\mu'\coloneqq(1/n+1/n^5)c\cdot T\cdot\gamma\cdot(n-n_1)$.
For $i\in\intcc{k}$ define
\begin{align*}
\ell_{\bot}(i) &\coloneqq \mu_i+\sqrt{c^2\cdot\log n\cdot T\cdot\gamma\frac{n_i}{n}}
\qquad\qquad\qquad\text{and}\\
\ell_{\top}    &\coloneqq T\gamma- \mu'-\sqrt{c^2\cdot\log n\cdot T\cdot\gamma\frac{n-n_1}{n}}
.
\end{align*}
We set $\ell_{\bot}\coloneqq\ell_{\bot}(2)$.
It is easy to show that $\ell_{\bot}<\ell_{\top}$.
Now, let all $\gamma n$ tokens be labeled from $1$ to $\gamma n$.
It remains to prove the lemma's statements:
\begin{enumerate}
\item For the first statement, consider a node $w$ with $\opinion{w}\geq2$ and set $\lambda(\opinion{w})\coloneqq\ell_{\bot}(\opinion{w})-\mu_{\opinion{w}}=\sqrt{c^2\cdot \log n\cdot T\cdot\gamma\cdot n_{\opinion{w}}/n}$.
	Set the random indicator variable $X_{i,t}$ to be $1$ if and only if $i$ is on node $w$ at time $t$ and if $i$'s label is $\opinion{w}$.
	Let $\cnt{w}=\sum_{i\in\intcc{\gamma n}}\sum_{j\leq T}X_{i,j\cdot\tmix}$.
	We compute
	\begin{equation}
	\begin{aligned}
	\Pr{\cnt{w}\geq\ell_{\bot}}
	&\leq \Pr{\cnt{w}\geq\mu_{\opinion{w}}+\lambda(\opinion{w})}\\
	&=    \Pr{\cnt{w}\geq\left(1+\frac{\lambda(\opinion{w})}{\mu_{\opinion{w}}}\right)\cdot\mu_{\opinion{w}}}\\
	&\leq \exp\left(-\frac{\lambda^2(\opinion{w})}{3\mu_{\opinion{w}}}\right)\leq\exp\left(-\frac{c}{6}\log n\right),
	\end{aligned}
	\end{equation}
	where the last line follows by Lemma~\ref{lem:chernoff_for_shuffle} applied to $\cnt{w}=\sum_{i\in\intcc{\gamma n}}\sum_{j\leq T}X_{i,j\cdot\tmix}$ and setting $B$ to the set of all tokens with label $\opinion{w}$.
	Hence, the claim follows for $c$ large enough after taking the union bound over all $n-n_1\leq n$ nodes $w$ with $\opinion{w}\geq2$.
\item For the lemma's second statement, consider a node $v$ with $\opinion{v}=1$ and set $\lambda'\coloneqq \mu'-\ell_{\top}$.
	Define the random indicator variable $Y_{i,t}$ to be $1$ if and only if token $i$ is on node $v$ at time $t$ and if $i$'s label is not $1$.
	Set $Y=\sum_{j\leq T}\sum_{i\in\intcc{\gamma n}}Y_{i,j\cdot\tmix}$ and note that $\cnt{v}=T\gamma-Y$.
	We compute
	\begin{align*}
	\Pr{\cnt{v}\leq\ell_{\top}}
	&=    \Pr{T\gamma - Y\leq\ell_{\top}}\\
	&=    \Pr{T\gamma-Y\leq T\gamma -\mu'-\lambda'}=\Pr{Y\geq \mu'+\lambda'}\\
	&=    \Pr{Y\geq\left(1+\frac{\lambda'}{\mu'}\right)\cdot \mu'}\leq\exp\left(-\frac{\lambda'^2}{3\mu'}\right)\\
	&\leq \exp\left(\frac{c}{6}\log n\right),
	\end{align*}
	where the first inequality follows by Lemma~\ref{lem:chernoff_for_shuffle} applied to $Y$ and using $B$ to denote the set of all tokens with a label other than $1$.
	Hence, the claim follows for $c$ large enough after taking the union bound over all $n_1\leq n$ nodes $v$ with $\opinion{u}\geq2$.
	\qedhere
\end{enumerate}
\end{proof}
With Lemma~\ref{onetwothree}, we can now prove our main result:
\begin{proof}[Proof of Theorem~\ref{thm:main}]
Fix an arbitrary time $t\in\intcc{c\cdot T\cdot\tmix,N}$ with $\tmix\mid t$, where $c$ is the constant from the statement of Lemma~\ref{onetwothree}.
From Lemma~\ref{onetwothree} we have that (w.h.p.) the node $u$ with the highest counter $\cnt{u}$ has $\opinion{u}=1$ (ties are broken arbitrarily).
In the following we condition on $\opinion{u}=1$.
We claim that at time $t'=t+\tmix$ all nodes $v\in V$ have $\majest{v}=1$.
This is because the counters during the \enquote{broadcast part} (Lines~\ref{shufflesimple:line:broadcast_start} to~\ref{shufflesimple:line:broadcast_stop}) propagate the highest counter received after time $t$.
The time $\tau$ until all nodes $v\in V$ have $\majest{v}=1$ is bounded by the mixing by definition:
In order for $\intcc{t,t'}$ to be $1/n^5$-smoothing, the random walk starting at $u$ at time $t$ is with probability at least $1/n-1/n^5$ on node $v$ and, thus, there exists a path from $u$ to $v$ (with respect to the communication matrices).
If there is such a path for every node $v$, the counter of $u$ was also propagated to that $v$ and we have $\tau\leq\tmix$.
Consequently, at time $t'$ all nodes have the correct majority opinion.
This implies the desired time bound.
For the memory requirements, note that each node $u$ stores $\gamma$ tokens with a label from the set $\intcc{k}$ ($\gamma\cdot\LDAUOmicron{\log k}$ bits), three opinions (its own, its plurality guess, and the dominating opinion; $\LDAUOmicron{\log k}$ bits),  the two counters $\cnt{u}$ and $\estimate{u}$ and the time step counter.
The memory to store the counter $\cnt{u}$ and $\estimate{u}$  is $\LDAUOmicron{\gamma T}$.
Finally, the time step counter is bounded by $\LDAUOmicron{\log(T\cdot \tmix)}$ bits.
Note that it is easy to implement a rolling counter when this counter \enquote{overflows}.
This yields the claimed space bound.
\end{proof}

\subsubsection*{Majorizing \shuffle by Random Walks}
We now turn to the proof of Lemma~\ref{lem:chernoff_for_shuffle}.
While our \shuffle protocol assumes that $2\Delta$ divides $\gamma$, we assume here the slightly weaker requirement that $\rw_t[u,v]\cdot\gamma\in\N$ for any $u,v\in V$ and $t\in\N$.
To ease the discussion, we consider $u$ as a neighbor of itself and speak of $d_t(u)+1$ neighbors.
For $i\in\intcc{d_t(u)+1}$, let $N_t(u,i)\in V$ denote the $i$-th neighbor of $u$ (in an arbitrary order).
We also need some notation for the shuffle part of our protocol.
To this end, consider a node $u$ at time $t$ and let $u$'s tokens be numbered from $1$ to $\gamma$.
Our assumption on $\gamma$ allows us to partition the tokens into $d_t(u)+1$ disjoint subsets (\emph{slots}) $S_i\subseteq\intcc{\gamma}$ of size $\rw_t[u,v]\cdot\gamma$ each, where $v=N_t(u,i)$.
Let $\pi_{t,u}\colon\intcc{\gamma}\to\intcc{\gamma}$ be a random permutation.
Token $j$ with $\pi_{t,u}(j)\in S_i$ is sent to $u$'s $i$-th neighbor.

Let $\mathscr{S}$ denote our random \shuffle process and $\mathscr{W}$ the random walk process in which each of the $\gamma n$ tokens performs an independent random walk according to the random walk matrices $(\rw_t)_{t\in\N}$.
We use $w_j^{\mathscr{P}}(t)$ to denote the position of token $j$ after $t$ steps of a process $\mathscr{P}$.
Without loss of generality, we assume $w_j^{\mathscr{S}}(0)=w_j^{\mathscr{W}}(0)$ for all tokens $j$.
While there are strong correlations between the tokens' movements in $\mathscr{S}$ (e.g., not all tokens can move to the same neighbor), Lemma~\ref{lem:diff_majorized_by_rw} shows that these correlations are negative.
Before proving Lemma~\ref{lem:diff_majorized_by_rw} we present the following definitions and auxilary results that are used in its proof.

\begin{definition}[{Neg.~Regression~\cite[Def.~21]{DBLP:journals/rsa/DubhashiR98}}]\label{def:neg_regression}
A vector $(X_1,X_2,\dots,X_n)$ of random variables is said to satisfy the \emph{negative regression} condition if $\Ex{f(X_l,l\in\mathscr{L})|X_r=x_r,r\in\mathscr{R}}$ is non-increasing in each $x_r$ for any disjoint $\mathscr{L},\mathscr{R}\subseteq\intcc{n}$ and for any non-decreasing function $f$.
\end{definition}

\begin{lemma}[{\cite[Lemma~26]{DBLP:journals/rsa/DubhashiR98}}]\label{lem:neg_regression_exp_inequality}
Let $(X_1,X_2,\dots,X_n)$ satisfy the negative regression condition and consider an arbitrary index set $I\subseteq\intcc{n}$ as well as any family of non-decreasing functions $f_i$ ($i\in\set I$).
Then, we have
\begin{equation}
\Ex{\prod_{i\in I}f_i(X_i)}\leq\prod_{i\in I}\Ex{f_i(X_i)}
\end{equation}
for any $I\subseteq\intcc{n}$ and for any non-decreasing functions $f_i$.
\end{lemma}
\begin{claim}\label{clm:pr_exp_identities}
Fix a time $t'\in\set{0,1,\dots,t-1}$ and consider an arbitrary configuration $c$.
Let $\mathscr{E}_{t'}$, $X_j$ and $h_j$ be defined as in the proof of Lemma~\ref{lem:diff_majorized_by_rw}.
Then the following identities hold:
\begin{enumerate}
\item\label{clm:pr_exp_identities:a} $\Pr{\mathscr{E}_{t'+1}|c(t')=c}=\Ex{\prod_{j\in B}h_j(X_j)|c(t')=c}$, and
\item\label{clm:pr_exp_identities:b} $\Pr{\mathscr{E}_{t'}|c(t')=c}=\prod_{j\in B}\Ex{h_j(X_j)|c(t')=c}$.
\end{enumerate}
\end{claim}
\begin{proof}
We use the shorthand $d(u_j)=d_{t'+1}(u_j)$.
Remember that each $X_j$ indicates to which of the $d(u_j)+1$ neighbors of $u_j$ (where $u_j$ is considered a neighbor of itself) a token $j$ moves during time step $t'+1$.
Thus, given the configuration $c(t')=c$ immediately before time step $t'+1$, there is a bijection between any possible configuration $c(t'+1)$ and outcomes of the random variable vector $\bm{X}=(X_j)_{j\in\intcc{\gamma n}}$.
Let $c_{\bm{x}}$ denote the configuration corresponding to a concrete outcome $\bm{X}=\bm{x}\in\intcc{d(u_j)+1}^{\gamma n}$.
Thus, we have $\Pr{c(t'+1)=c_{\bm{x}}|c(t')=c}=\Pr{\bm{X}=\bm{x}|c(t')=c}$, and conditioning on $c(t'+1)$ is equivalent to conditioning on $\bm{X}$ and $c(t')$.
For the claim's first statement, we calculate
\begin{equation*}
\begin{aligned}
      &\Pr{\mathscr{E}_{t'+1}|c(t')=c}\\
{}={} &\sum_{c_{\bm{x}}}\Pr{\mathscr{E}_{t'+1}|c(t'+1)=c_{\bm{x}}}\cdot\Pr{c(t'+1)=c_{\bm{x}}|c(t')=c}\\
{}={} &\sum_{c_{\bm{x}}}\prod_{j\in B}\Pr{w_j^{\mathscr{SW}(t'+1)}(t)\in D|\bm{X}=\bm{x},c(t')=c}\cdot\Pr{\bm{X}=\bm{x}|c(t')=c}\\
{}={} &\sum_{c_{\bm{x}}}\prod_{j\in B}h_j(x_j)\cdot\Pr{\bm{X}=\bm{x}|c(t')=c}\\
{}={} &\sum_{\bm{x}}\prod_{j\in B}h_j(x_j)\cdot\Pr{\bm{X}=\bm{x}|c(t')=c}=\Ex{\prod_{j\in B}h_j(X_j)|c(t')=c}
.
\end{aligned}
\end{equation*}
Here, we first apply the law of total probability.
Then, we use the bijection between $c(t'+1)$ and $\bm{X}$ (if $c(t')$ is given) and that the process $\mathscr{SW}(t'+1)$ consists of independent random walks if $c(t'+1)$ is fixed.
Finally, we use the definition of the auxiliary functions $h_j(i)$, which equal the probability that a random walk starting at time $t'+1$ from $u_j$'s $i$-th neighbor reaches a node from $D$.

For the claim's second statement, we do a similar calculation for the process $\mathscr{SW}(t')$.
By definition, this process consists already from time $t'$ onwards of a collection of independent random walks.
Thus, we can swap the expectation and the product in the last term of the above calculation, yielding the desired result.
\end{proof}

\begin{lemma}\label{lem:diff_majorized_by_rw}
Consider a time $t\geq0$, a token $j$, and node $v$.
Moreover, let $B\subseteq\intcc{\gamma n}$ and $D\subseteq V$ be arbitrary subset of tokens and nodes, respectively.
Then, the following holds:
\begin{enumerate}
\item $\Pr{w_j^{\mathscr{S}}(t)=v}=\Pr{w_j^{\mathscr{W}}(t)=v}$ and
\item $\Pr{\bigcap_{j\in B}\left(w_j^{\mathscr{S}}(t)\in D\right)}\leq\Pr{\bigcap_{j\in B}\left(w_j^{\mathscr{W}}(t)\in D\right)} \\ =\prod_{j\in B}\Pr{w_j^{\mathscr{W}}(t)\in D}$.
\end{enumerate}
\end{lemma}
\begin{proof}
The first statement follows immediately from the definition of our process.
For the second statement, note that the equality on the right-hand side holds trivially, since the tokens perform independent random walks in $\mathscr{W}$.
To show the inequality, we define intermediate processes $\mathscr{SW}(t')$ ($t'\leq t$) that perform $t'$ steps of $\mathscr{S}$ followed by $t-t'$ steps of $\mathscr{W}$.
By this definition, $\mathscr{SW}(0)$ is identical to $\mathscr{W}$ restricted to $t$ steps and, similar, $\mathscr{SW}(t)$ is identical to $\mathscr{S}$ restricted to $t$ steps.
Consider the events
\begin{equation}
\mathscr{E}_{t'}\coloneqq\bigcap_{j\in B}\left(w_j^{\mathscr{SW}(t')}(t)\in D\right)
,
\end{equation}
stating that all tokens from $B$ end up at nodes from $D$ under process $\mathscr{SW}(t')$.
The lemma's statement is equivalent to $\Pr{\mathscr{E}_t}\leq\Pr{\mathscr{E}_0}$.
To prove this, we show $\Pr{\mathscr{E}_{t'+1}}\leq\Pr{\mathscr{E}_{t'}}$ for all $t'\in\set{0,1,\dots,t-1}$.
Combining these inequalities yields the desired result.

Fix an arbitrary $t'\in\set{0,1,\dots,t-1}$ and note that $\mathscr{SW}(t')$ and $\mathscr{SW}(t'+1)$ behave identical up to and including step $t'$.
Hence, we can fix an arbitrary configuration (i.e., the location of each token) $c(t')=c$ immediately before time step $t'+1$.
For a token $j\in\intcc{\gamma n}$ let $u_j\in V$ denote its location in configuration $c$.
Remember that $\pi_{u,t'+1}$ denote the (independent) random permutations chosen by each node $u$ for time step $t'+1$.
To ease notation, we drop the time index $t'+1$ and write $\pi_u$ instead of $\pi_{u,t'+1}$ (and, similarly, $d(u)$ and $N(u,i)$ instead of $d_{t'+1}(u)$ and $N_{t'+1}(u,i)$).
For each token $j$ define a random variable $X_j\in\intcc{d(u_j)+1}$ with $X_j=i$ if and only if $\pi_{u_j}(j)\in S_i$.
In other words, $X_j$ indicates to which of $u_j$'s neighbors token $j$ is sent in time step $t'+1$.
We also introduce auxiliary functions $h_j\colon\intcc{d(u_j)+1}\to\intcc{0,1}$ defined by
\begin{equation}
h_j(i)\coloneqq\Pr{w_j^{\mathscr{W}}(t)\in D|w_j^{\mathscr{W}}(t'+1)=N(u_j,i)}
.
\end{equation}
These are the probability that a random walk starting at time $t'+1$ from $u_j$'s $i$-th neighbor ends up in a node from $D$.
We can assume (w.l.o.g.) that all $h_j$ are non-decreasing (by reordering the neighborhood of $u_j$ accordingly).
We show in Lemma~\ref{lem:neg_regression} that the variables $(X_j)_{j\in B}$ satisfy the negative regression condition (cf.~Definition~\ref{def:neg_regression}).
Additionally, Claim~\ref{clm:pr_exp_identities} relates the (conditioned) probabilities of the events $\mathscr{E}_{t'}$ and $\mathscr{E}_{t'+1}$ to the expectations over the different $h_j(X_j)$.
With this, we get
\begin{equation*}
\begin{aligned}
\Pr{\mathscr{E}_{t'+1}|c(t')=c}
&\stackrel{Clm.~\ref{clm:pr_exp_identities}\ref{clm:pr_exp_identities:a}}{=}\Ex{\prod_{j\in B}h_j(X_j)|c(t')=c}\\
&\disguisemath{\stackrel{Clm.~\ref{clm:pr_exp_identities}\ref{clm:pr_exp_identities:a}}{=}{}}{\stackrel{Lem.~\ref{lem:neg_regression_exp_inequality}}{\leq}}\prod_{j\in B}\Ex{h_j(X_j)|c(t')=c}\\
&\stackrel{Clm.~\ref{clm:pr_exp_identities}\ref{clm:pr_exp_identities:b}}{=}\Pr{\mathscr{E}_{t'}|c(t')=c}
.
\end{aligned}
\end{equation*}
Using the law of total probability, we conclude $\Pr{\mathscr{E}_{t'+1}}\leq\Pr{\mathscr{E}_{t'}}$, as required.
\end{proof}

\begin{lemma}\label{lem:neg_regression}
Fix a time $t'<t$ and an arbitrary configuration $c$.
Let $X_j$ be as in the proof of Lemma~\ref{lem:diff_majorized_by_rw}.
Then the vector $(X_j)_{j\in\intcc{\gamma n}}$ satisfies the negative regression condition (\NRC).
\end{lemma}
\begin{proof}
Remember that $u_j$ is the location of token $j$ in configuration $c$ and that $X_j\in\intcc{d_{t'+1}(u_j)+1}$ indicates where token $j$ is sent.
We show for any $u\in V$ that $(X_j)_{j\colon u_j=u}$ satisfies the \NRC.
The lemma's statement follows since the $\pi_{u}$ are chosen independently (if two independent vectors $(X_j)$ and $(Y_j)$ satisfy the \NRC, then so do both together.

Fix a node $u$ and disjoint subsets $\mathscr{L},\mathscr{R}\subseteq\set{j\in\intcc{\gamma n}|u_j=u}$ of tokens on $u$.
Define $d\coloneqq d_{t'+1}(u)$ and let $f\colon\intcc{d+1}^{\abs{\mathscr{L}}}\to\R$ be an arbitrary non-decreasing function.
We have to show that $\Ex{f(X_l,l\in\mathscr{L})|X_r=x_r,r\in\mathscr{R}}$ is non-increasing in each $x_r$ (cf.~Definition~\ref{def:neg_regression}).
That is, we need
\begin{equation}\label{eqn:negreg_via_coupling}
\Ex{f(X_l,l\in\mathscr{L})|X_r=x_r,r\in\mathscr{R}}\leq\Ex{f(X_l,l\in\mathscr{L})|X_r=\tilde{x}_r,r\in\mathscr{R}}
,
\end{equation}
where $x_r=\tilde{x}_r$ holds for all $r\in\mathscr{R}\setminus\set{\hat{r}}$ and $x_{\hat{r}}>\tilde{x}_{\hat{r}}$ for a fixed index $\hat{r}\in\mathscr{R}$.
We prove Inequality~\eqref{eqn:negreg_via_coupling} via a coupling of the processes on the left-hand side (LHS process) and right-hand side (RHS process) of that inequality.
Since $x_{\hat{r}}\neq\tilde{x}_{\hat{r}}$, these processes involve two slightly different probability spaces $\Omega$ and $\tilde{\Omega}$, respectively.
To couple these, we employ a common uniform random variable $U_i\in\intco{0,1}$.
By partitioning $\intco{0,1}$ into $d+1$ suitable slots for each process (corresponding to the slots $S_i$ from the definition of $\mathscr{S}$), we can use the outcome of $U_i$ to set the $X_j$ in both $\Omega$ and $\tilde{\Omega}$.
We first explain how to handle the case $x_{\hat{r}}-\tilde{x}_{\hat{r}}=1$.
The case $x_{\hat{r}}-\tilde{x}_{\hat{r}}>1$ follows from this by a simple reordering argument.

So assume $x_{\hat{r}}-\tilde{x}_{\hat{r}}=1$.
We reveal the yet unset random variables $X_j$ (i.e., $j\in\intcc{\gamma n}\setminus\mathscr{R}$) one by one in order of increasing indices.
To ease the description assume (w.l.o.g.) that the tokens from $\mathscr{R}$ are numbered from $1$ to $\abs{\mathscr{R}}$.
When we reveal the $j$-th variable (which indicates the new location of the $j$-th token), note that the probability $p_{j,i}$ that token $j$ is assigned to $N(u,i)$ depends solely on the \emph{number} of previous tokens $j'<j$ that were assigned to $N(u,i)$.
Thus, we can consider $p_{j,i}\colon\N\to\intcc{0,1}$ as a function mapping $x\in\N$ to the probability that $j$ is assigned to $N(u,i)$ conditioned on the event that exactly $x$ previous tokens were assigned to $N(u,i)$.
Note that $p_{j,i}$ is non-increasing.
For a vector $\bm{x}\in\N^{d+1}$, we define a threshold function $T_{j,i}\colon\N^{d+1}\to\intcc{0,1}$ by $T_{j,i}(\bm{x})\coloneqq\sum_{i'\leq i}p_{j,i'}(x_{i'})$ for each $i\in\intcc{d+1}$.
To define our coupling, let $\beta_{j,i}\coloneqq\abs{\set{j'<j|X_{j'}=i}}$ denote the number of already revealed variables with value $i$ in the LHS process and define, similarly, $\tilde{\beta}_{j,i}\coloneqq\abs{\set{j'<j|\tilde{X}_{j'}=i}}$ for the RHS process.
We use $\bm{\beta_j},\bm{\tilde{\beta}_j}\in\N^{d+1}$ to denote the corresponding vectors.
Now, to assign token $j$ we consider a uniform random variable $U_j\in\intco{0,1}$ and assign $j$ in both processes using customized partitions of the unit interval.
To this end, let $T_{j,i}\coloneqq T_{j,i}(\bm{\beta_j})$ and $\tilde{T}_{j,i}\coloneqq T_{j,i}(\bm{\tilde{\beta}_j})$ for each $i\in\intcc{d+1}$.
We assign $X_j$ in the LHS and RHS process as follows:
\begin{itemize}
\item \textbf{LHS Process:} $X_j=x_j=i$ if and only if $U_j\in\intco{T_{j,i-1},T_{j,i}}$,
\item \textbf{RHS Process:} $X_j=\tilde{x}_j=i$ if and only if $U_j\in\intco{\tilde{T}_{j,i-1},\tilde{T}_{j,i}}$.
\end{itemize}
See Figure~\ref{fig:neg_regression_coupling} for an illustration.
Our construction guarantees that, considered in isolation, both the LHS and RHS process behave correctly.

\begin{figure}
\includegraphics[width=\linewidth]{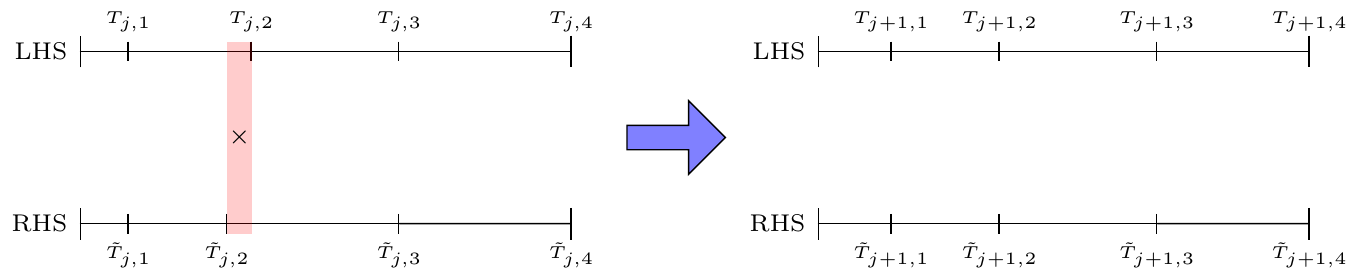}
\caption{%
	Illustration showing the $d+1=4$ different slots for the LHS and RHS process and how they change.
	In this example, $x_{\hat{r}}=3$ and $\tilde{x}_{\hat{r}}=2$.
	On the left, the uniform random variable $U_j$ falls into slot $\intco{T_1,T_2}$ for the LHS process (causing $j$ to be sent to node $N(u,2)$) and into slot $\intco{\tilde{T}_2,\tilde{T}_3}$ for the RHS process (causing $j$ to be sent to node $N(u,3)$).}
\label{fig:neg_regression_coupling}
\end{figure}

At the beginning of this coupling, only the variables $X_r$ corresponding to tokens $r\in\mathscr{R}$ are set, and these differ in the LHS and RHS process only for the index $\hat{r}\in\mathscr{R}$, for which we have $X_{\hat{r}}=x_{\hat{r}}$ (LHS) and $X_{\hat{r}}=\tilde{x}_{\hat{r}}=x_{\hat{r}}-1$ (RHS).
Let $\iota\coloneqq x_{\hat{r}}$.
For the first revealed token $j=\hat{r}+1$, this implies $\beta_{j,\iota}=\tilde{\beta}_{j,\iota}+1$, $\beta_{j,\iota-1}=\tilde{\beta}_{j,\iota-1}-1$, and $\beta_{j,i}=\tilde{\beta}_{j,i}$ for all $i\not\in\set{\iota,\iota-1}$.
By the definitions of the slots for both processes, we get $T_{j,i}=\tilde{T}_{j,i}$ for all $i\neq\iota-1$ and $T_{j,\iota-1}>\tilde{T}_{j,\iota-1}$ (cf.~Figure~\ref{fig:neg_regression_coupling}).
Thus, the LHS and RHS process behave different if and only if $U_i\in\intco{\tilde{T}_{j,\iota-1},T_{j,\iota-1}}$.
If this happens, we get $x_j<\tilde{x}_j$ (i.e., token $j$ is assigned to a smaller neighbor in the LHS process).
This implies $\bm{\beta_{j+1}}=\bm{\tilde{\beta}_{j+1}}$ and both processes behave identical from now on.
Otherwise, if $U_i\not\in\intco{\tilde{T}_{j,\iota-1},T_{j,\iota-1}}$, we have $\bm{\beta_{j+1}}-\bm{\beta_{j+1}}=\bm{\beta_j}-\bm{\beta_j}$ and we can repeat the above argument.
Thus, after all $X_j$ are revealed, there is at most one $j\in\mathscr{L}$ for which $x_j\neq\tilde{x}_j$, and for this we have $x_j<\tilde{x}_j$.
Since $f$ is non-decreasing, this guarantees Inequality~\eqref{eqn:negreg_via_coupling}.
To handle the case $x_{\hat{r}}-\tilde{x}_{\hat{r}}>1$, note that we can reorder the slots $\intco{T_{j,i-1},T_{j,i}}$ used for the assignment of the variables such that the slots for $x_{\hat{r}}$ and $\tilde{x}_{\hat{r}}$ are neighboring.
Formally, this merely changes in which order we consider the neighbors in the definition of the functions $T_{j,i}$.
With this change, the same arguments as above apply.
\end{proof}
With the above tools, we finally are able to prove the Chernoff bound stated in Lemma~\ref{lem:chernoff_for_shuffle}.
\begin{proof}[Proof of Lemma~\ref{lem:chernoff_for_shuffle}]
Let $v_{j,t}$ denote the location of token $j$ at time $(t-1)\cdot\tmix$.
For all $t\leq T$ and $\ell\in\N$ define the random indicator variable $Y_{j,t}$ to be $1$ if and only if the random walk starting at $v_{j,t}$ is at node $u$ after $\tmix$ time steps.
By Lemma~\ref{lem:diff_majorized_by_rw} we have for each $B'\subseteq B$ and $t\leq T$ that
\begin{equation}
\Pr{\bigcap_{i\in B'}X_{j,t}=1}\leq\prod_{j\in B'}\Pr{Y_{j,t}=1}
.
\end{equation}
Hence for all $t\leq T$ and $\ell\in\N$ we have $\Pr{\sum_{j\in B}X_{j,t}\geq\ell}\leq\Pr{\sum_{j\in B}Y_{j,t}\geq\ell}$ and
\begin{equation}
\Pr{X\geq\ell}=\Pr{\sum_{t\leq T}\sum_{j\in B}X_{j,t}\geq\ell}\leq\Pr{\sum_{t\leq T}\sum_{j\in B}Y_{j,t}\geq\ell}
.
\end{equation}
Let us define $p\coloneqq1/n+1/n^5$.
By the definition of $\tmix$, we have for all $j\in B$ and $t\leq T$ that
\begin{equation}
\Pr{Y_{j,t}=1|Y_{1,1},Y_{2,1},\dots,Y_{\abs{B},1},Y_{1,2},\dots,Y_{j-1,t}}\leq p
.
\end{equation}
Combining our observations with Lemma~\ref{BoundedChernoff} (see below), we get $\Pr{X\geq\ell}\leq\Bin(T\cdot\abs{B},p)$.
Recall that $\mu=T\cdot\abs{B}\cdot p$.
Thus, by applying standard Chernoff bounds we get
\begin{equation}
\Pr{X\geq(1+\delta)\mu}\leq\left(\frac{e^{\delta}}{(1+\delta)^{1+\delta}}\right)^{\mu}\leq e^{\delta^2\mu/3}
,
\end{equation}
which yields the desired statement.
\end{proof}
\begin{lemma}[{\cite[Lemma~3.1]{DBLP:journals/siamcomp/AzarBKU99}}]\label{BoundedChernoff}
Let $X_1,X_2,\dots,X_n$ be a sequence of random variables with values in an arbitrary domain and let $Y_1,Y_2,\dots,Y_n$ be a sequence of binary random variables with the property that $Y_i=Y_i(X_1,\dots,X_i)$.
If $\Pr{Y_i=1|X_1,\dots,X_{i-1}}\leq p$, then
\begin{align}
\Pr{\sum Y_i\geq\ell}\leq\Pr{\Bin(n,p)\geq\ell}
\end{align}
and, similarly, if $\Pr{Y_i=1|X_1,\dots,X_{i-1}}\geq p$, then
\begin{align}
\Pr{\sum Y_i\leq\ell}\leq\Pr{\Bin(n,p)\leq\ell}
.
\end{align}
Here, $\Bin(n,p)$ denotes the binomial distribution with parameters $n$ and $p$.
\end{lemma}


\section{Protocol \balance}\label{sec:balance}
\myparagraph{Protocol Description}
The idea of our \balance protocol is quite simple:
Every node $u$ stores a $k$-dimensional vector $\bm{\ell_t(u)}$ with $k$ integer entries, one for each opinion.
\balance simply performs an entry-wise load balancing on $\bm{\ell_t(u)}$ according to the communication pattern $\mat=(\mat_t)_{t\leq N}$ and the corresponding transition matrices $\rw_t$ (cf.~Section~\ref{subsec:random_walk}).
Once the load is properly balanced, the nodes look at their largest entry and assume that this is the plurality opinion (stored in the variable $\majest{u}$).

In order to ensure a low memory footprint, we must not send fractional loads over the active edges.
To this end, we use a rounding scheme from~\cite{DBLP:journals/jcss/BerenbrinkCFFS15,DBLP:conf/focs/SauerwaldS12}, which works as follows:
Consider a dimension $i\in\intcc{k}$ and let $\ell_{i,t}(u)\in\N$ denote the current (integral) load at $u$ in dimension $i$.
Then $u$ sends $\floor{\ell_{i,t}(u)\cdot\rw_t[u,v]}$ tokens to all neighbors $v$ with $\mat_t[u,v]=1$.
This results in at most $d_t(u)$ remaining \emph{excess tokens} ($\ell_{i,t}(u)$ minus the total number of tokens sent out).
These are then randomly distributed (without replacement), where neighbor $v$ receives a token with probability $\rw_t[u,v]$.
In the following we call the resulting balancing algorithm \emph{vertex-based balancing} algorithm.
The formal description of protocol \balance is given in Listing~\ref{alg:balance}.

\begin{lstlisting}[float,label={alg:balance},caption={%
	Protocol \balance as executed by node $u$ at time $t$.
	At time zero, each node initializes $\ell_{o_u,0}(u)\coloneqq\gamma$ and $\ell_{j,0}(u)\coloneqq0$ for all $j\neq o_u$.
}]
for $i\in\intcc{k}$:
  for $\{u,v\}\in E$ with $\mat_t[u,v]=1$:
    send $\floor{{\ell_{i,t}(u)\cdot\rw_t[u,v]}}$ tokens from dimension $i$ to $v$
  $x\coloneqq\ell_{i,t}(u)-\sum_{v\colon\mat_t[u,v]=1}\floor{{\ell_{i,t}(u)\cdot\rw_t[u,v]}}$       {excess tokens}
  randomly distribute $x$ tokens such that:
    every $v\neq v$ with $\mat_t[u,v]=1$ receives $1$ token w.p. $\rw_t[u,v]$
    (and zero otherwise)
$\majest{u}\coloneqq i$ with $\ell_{i,t}(u)\geq\ell_{j,t}(u)\quad\forall 1\leq i,j\leq k$            {plurality guess}
\end{lstlisting}

\myparagraph{Analysis of \balance}
Consider initial load vectors $\bm{\ell_0}$ with $\norm{\bm{\ell_0}}_{\infty}\leq n^5$.
Let $\tau\coloneqq\tau(g,\mat)$ be the first time step when \vbased under the (fixed) communication pattern $\mat=(\mat_t)_{t\leq N}$ is able to balance any such vector $\bm{\ell_0}$ up to a $g$-discrepancy (i.e., the minimal $t$ with $\disc(\bm{\ell_t})\leq g$).
With these definitions, one can easily prove the following theorem.
\begin{theorem}\label{thm:parallelbalancing}
Let $\alpha = \frac{n_1-n_2}{n}\in\intcc{\sfrac{1}{n},1}$ denotes the initial bias.
Consider a fixed communication pattern $\mat=(\mat_t)_{t\leq N}$ and an integer $\gamma\in\intcc{3\cdot\frac{g}{\alpha},n^5}$.
Protocol \balance ensures that all nodes know the plurality opinion after $\tau(g,\mat)$ rounds and requires $k\cdot\log(\gamma)$ memory bits per node.
\end{theorem}
\begin{proof}
Recall that $\gamma\geq3\frac{g}{\alpha}=3g\cdot\frac{n}{n_1-n_2}$.
For $i\in\intcc{k}$ let $\bar\ell_i\coloneqq n_i\cdot\gamma/n$.
The definition of $\tau(g,\mat)$ implies $\ell_{1,t}(u)\geq\bar\ell_1-g$ and $\ell_{i,t}(u)\leq\bar\ell_i+g$ for all nodes $u$ and $i\geq2$.
Consequently, we get
\begin{equation}
\ell_{1,t}(u)-\ell_{i,t}(u)\geq\bar\ell_1-\bar\ell_i-2g=3g\cdot\frac{n_1-n_i}{n_1-n_2}-2g>0
.
\end{equation}
Thus, every node $u$ has the correct plurality guess at time $t$.
\end{proof}
The memory usage of \balance depends on the number of opinions ($k$) and on the number of tokens generated on every node ($\gamma$).
The algorithm is very efficient for small values of $k$ but it becomes rather impractical if $k$ is large.
Note that if one chooses $\gamma$ sufficiently large, it is easy to adjust the algorithm such that every node knows the frequency of \emph{all} opinions in the network.
The following corollary gives a few concrete examples for common communication patterns on general graphs.


\begin{corollary}\label{mastercor} Let $G$ be an arbitrary $d$-regular graph.
\balance ensures that all nodes agree on the plurality opinion with probability $1-e^{-(\log(n))^c}$ for some constant $c$
\begin{enumerate}
\item\label{mastercor:a} using $\LDAUOmicron{k\cdot\log n}$ bits of memory in time $\LDAUOmicron[small]{\frac{\log n}{1-\lambda_2}}$ in the diffusion model,
\item\label{mastercor:b} using $\LDAUOmicron{k\cdot\log n}$ bits of memory in time $\LDAUOmicron[small]{\frac{1}{d\cdot\pmin}\cdot\frac{\log n}{1-\lambda_2}}$ in the random matching model,
\item\label{mastercor:c} using $\LDAUOmicron{k\cdot\log(\alpha^{-1})}$ bits of memory in time $\LDAUOmicron[small]{d\cdot\frac{\log n}{1-\lambda_2}}$ in the balancing circuit model, and
\item\label{mastercor:d} using $\LDAUOmicron{k\cdot\log(\alpha^{-1})}$ bits of memory in time $\LDAUOmicron[small]{n\cdot\frac{\log n}{1-\lambda_2}}$ in the sequential model.
\end{enumerate}
\end{corollary}
\begin{proof}
Part~\ref{mastercor:a} follows directly from~\cite[Theorem~6.6]{DBLP:journals/corr/abs-1201-2715} and Part~\ref{mastercor:c} follows directly from~\cite[Theorem~1.1]{DBLP:journals/corr/abs-1201-2715}.
To show Part~\ref{mastercor:b} and~\ref{mastercor:d} we choose $\tau$ such that $\mat_1,\mat_2,\dots,\mat_{\tau}$ enable \vbased to balance any vector $\bm{\ell_0}$ (with initial discrepancy of at most $n^5$) up to a $g$-\emph{discrepancy}.
The bound on $\tau$ then follows from~\cite[Theorem~1.1]{DBLP:journals/corr/abs-1201-2715}.
\end{proof}
In particular, for the complete graph and $k=2$, Corollary~\ref{mastercor}\ref{mastercor:d} gives the same bounds as the bound from~\cite{DBLP:conf/podc/AlistarhGV15}.
Note that the $s$ states used to measure space requirement in~\cite{DBLP:conf/podc/AlistarhGV15} correspond to $\log s$ memory bits in our model.


\bibliographystyle{plain}
\bibliography{dblp}

\end{document}